\documentclass[11pt]{article}
\usepackage{fullpage}
\usepackage{graphicx}
\usepackage{xcolor}
\definecolor{ForestGreen}{rgb}{0.1333,0.5451,0.1333}
\definecolor{Red}{rgb}{0.9,0,0}
\definecolor{DarkRed}{rgb}{0.1,0.1,0.9}
\definecolor{DarkBlue}{rgb}{0.1,0.1,0.5}

\usepackage{amsmath}
\usepackage{amssymb}
\usepackage{mathtools}
\usepackage{amsthm}
\usepackage{enumitem}
\theoremstyle{plain}

\usepackage[linktocpage=true,pagebackref=true,colorlinks,urlcolor=black,linkcolor=DarkRed,citecolor=ForestGreen,bookmarks,bookmarksopen,bookmarksnumbered]{hyperref}
\usepackage[noabbrev,nameinlink]{cleveref}

\newtheorem{theorem}{Theorem}[section]

\newtheorem{lemma}[theorem]{Lemma}
\newtheorem{corollary}[theorem]{Corollary}
\theoremstyle{definition}
\newtheorem{definition}[theorem]{Definition}

\newtheorem{claim}[theorem]{Claim}
\newtheorem{fact}[theorem]{Fact}
\theoremstyle{remark}

\usepackage{mdframed}
\mdfdefinestyle{problemstyle}{%
  linecolor=white,
  linewidth=1pt,
  frametitlerule=true,
  frametitlebackgroundcolor=gray!20,
  backgroundcolor=gray!10
}
\theoremstyle{definition}
\newmdtheoremenv[style=problemstyle]{problem}{Problem}
\usepackage[textsize=tiny]{todonotes}

\usepackage{algorithmic}
\usepackage{algorithm}

\newtheorem{mdalg}{Algorithm}

\newtheorem{mdimp}{Implementation}

\allowdisplaybreaks

\renewcommand{\qed}{\nobreak \ifvmode \relax \else
	\ifdim\lastskip<1.5em \hskip-\lastskip
	\hskip1.5em plus0em minus0.5em \fi \nobreak
	\vrule height0.75em width0.5em depth0.25em\fi}

\newcommand{\Qed}[1]{\ensuremath{\qed_{\textnormal{~#1}}}}

\newcommand{\Ot}{\ensuremath{\widetilde{O}}}
\newcommand{\eps}{\ensuremath{\varepsilon}}

\newcommand{\norm}[1]{\ensuremath{\|#1\|}}

\newcommand{\floor}[1]{{\left\lfloor{#1}\right\rfloor}}

\newcommand{\poly}{\mbox{\rm poly}}

\DeclareMathOperator*{\Prob}{\ensuremath{\textnormal{Pr}}}
\renewcommand{\Pr}{\Prob}

\newenvironment{tbox}{\begin{tcolorbox}[
		enlarge top by=5pt,
		enlarge bottom by=5pt,
		breakable,
		boxsep=0pt,
		left=4pt,
		right=4pt,
		top=10pt,
		arc=0pt,
		boxrule=1pt,toprule=1pt,
		colback=white
		]%%
	}
	{\end{tcolorbox}}

\newcommand{\supp}[1]{\ensuremath{\textnormal{\text{supp}}(#1)}}

\newcommand{\II}{\ensuremath{\mathbb{I}}}

\newcommand{\mireal}[1][]{
	\ifx\relax#1\relax%
	\II(\mione \,; \mitwo)%
	\else%
	\II(\mione \,; \mitwo\mid #1)%
	\fi
}

\crefname{algocfline}{alg.}{algs.}
\Crefname{algocfline}{Algorithm}{Algorithms}
\newcommand{\E}{\mathop{\mathbb{E}}}

\newcommand{\veca}{\ensuremath{\mathbf{a}}}
\newcommand{\vecb}{\ensuremath{\mathbf{b}}}

\newcommand{\vecm}{\ensuremath{\mathbf{m}}}

\newcommand{\vecq}{\ensuremath{\mathbf{q}}}

\newcommand{\vecu}{\ensuremath{\mathbf{u}}}
\newcommand{\vecv}{\ensuremath{\mathbf{v}}}

\newcommand{\vecx}{\ensuremath{\mathbf{x}}}
\newcommand{\vecy}{\ensuremath{\mathbf{y}}}
\newcommand{\vecz}{\ensuremath{\mathbf{z}}}

\let\oldnl\nl% Store \nl in \oldnl
\newcommand{\nonl}{\renewcommand{\nl}{\let\nl\oldnl}}% Remove line number for one line

\newcommand{\hypercube}{\{0,1\}^d}
\newcommand{\unitsphere}{S^{d-1}}
\newcommand{\hidden}{H}
\newcommand{\disth}
{\ensuremath{\text{dist}}}
\newcommand{\wt}{\text{wt}}

\newcommand{\query}[1]{\textnormal{query}(#1)}
\newcommand{\conv}{\textsc{Conv}}
\newcommand{\chull}[1]{\conv(#1)}
\newcommand{\spn}[1]{\text{Span}(#1)}
\newcommand{\support}{\text{Supp}}
\newcommand{\hgm}{\textup{Hypergeometric}}
\newcommand{\R}{\mathbb{R}}

\newcommand{\negate}[2]{#1_{\sim #2}}

\DeclareMathOperator*{\argmin}{argmin}

\newcommand{\deter}{\textsc{Det}}
\title{Learning Multiple Secrets in Mastermind}
\author{
Milind Prabhu \thanks{University of Michigan,  \url{milindpr@umich.edu}} \and
David Woodruff \thanks{Carnegie Mellon University, \url{dwoodruf@cs.cmu.edu}}
}
\begin{document}
\maketitle

\begin{abstract}

 In the Generalized Mastermind problem, there is an unknown subset $H$ of the hypercube $\{0,1\}^d$ containing $n$ points. The goal is to learn $H$ by making a few queries to an oracle, which, given a point  $\vecq$ in $\{0,1\}^d$, returns the point in $H$ nearest to $\vecq$. We give a two-round adaptive algorithm for this problem that learns $H$ while making at most $\exp(\widetilde{O}(\sqrt{d \log n}))$ queries. Furthermore, we show that any $r$-round adaptive randomized algorithm that learns $H$ with constant probability must make $\exp(\Omega(d^{3^{-(r-1)}}))$ queries even when the input has $\poly(d)$ points; thus, any $\poly(d)$ query algorithm must necessarily use $\Omega(\log \log d)$ rounds of adaptivity. We give optimal query complexity bounds for the variant of the problem where queries are allowed to be from $\{0,1,2\}^d$. We also study a continuous variant of the problem in which $\hidden$ is a subset of unit vectors in $\R^d$, and one can query unit vectors in $\R^d$. For this setting, we give an  $O(n^{\floor{d/2}})$ query deterministic algorithm to learn the hidden set of points.
\end{abstract}

\section{Introduction}
\label{sec: intro}
Mastermind is a classic codebreaking board game that originated in 1970 and has inspired several lines of research in theoretical computer science over the past few decades. The game is played by two players: a codemaker and a codebreaker. The codemaker chooses a secret sequence of four colored code pegs, each with one of six colors; the codebreaker who wishes to determine the sequence is then allowed to guess various four-peg sequences. For each sequence the codebreaker guesses, they learn the number of pegs in their guess that are of the correct color and appear in the correct position, as well as the number of pegs that are of the correct color but are in the wrong position. The codebreaker wishes to learn the sequence while making as few guesses as possible. 

Knuth \cite{Knuth1977TheCA} showed that the original Mastermind game has an optimal strategy that uses $5$ queries in the worst case. The generalization of the Mastermind game to $n$ pegs and $k$ colors was studied by Chv\'atal \cite{chvatal1983mastermind}, who gave an asymptotically optimal when $k \leq n^{1-\eps}$. The regime when $k\geq n$ was less well understood until recently.  For the specific case when $k = n$, Chv\'atal gave a strategy with $O(n \log n)$ query complexity, and several subsequent works \cite{chen1996finding, goodrich2009algorithmic,jager2009number} gave strategies with improved leading constants. The first asymptotic improvement in the query complexity to $O(n \log \log n)$ queries was obtained by  Doerr et al. \cite{doerr2016playing}. Recently, Martinsson and Su \cite{martinsson2024mastermind} obtained an $O(n)$ query strategy, thus matching the information-theoretic lower bound up to constant factors.

\subsection{Mastermind With Multiple Secrets}
In this paper, we introduce and study variants of the Mastermind problem in which the codemaker chooses not one hidden point but a set $\hidden$ of $n$ hidden points.  The codebreaker operates under a certain query model and has to learn \emph{all} the points in $\hidden$ while minimizing the number of queries they make. Specifically, we consider three variants of this problem, each defined by the space from which $\hidden$ is chosen and the query model under which the codebreaker operates.  In the first variant, we consider the setting of the problem when the hidden points lie in the hypercube.

\begin{problem}\label{prob: hamming-vanilla}
The hidden set $\hidden$ is a subset of the hypercube $\{0,1\}^d$ of size $n$. The codebreaker is allowed to query $\vecq \in \{0,1\}^d$ of their choice in response to which they learn $\argmin_{\vecx \in \hidden} \disth(\vecx, \vecq)$, i.e., the point in $\hidden$ with the least Hamming distance to $\vecq$.
\end{problem}
Next, we consider a continuous variant of the problem.
\begin{problem}\label{prob: unitsphere}
The hidden set $\hidden$ is a subset of the unit sphere $\unitsphere$ of size $n$. The codebreaker is allowed to query  $\vecq \in \unitsphere$ of their choice in response to which they learn $\argmin_{\vecx \in \hidden} \norm{\vecx - \vecq}_2$, i.e., the point in $\hidden$ that has the least Euclidean distance to $\vecq$.
\end{problem}
%While in \Cref{prob: hamming-vanilla}, the codebreaker is restricted to query points in the hypercube, we study the problem in which queries are allowed to have extra characters. 
We also study the hypercube version of the problem with a stronger query model, in which the codebreaker has access to queries with an extra character.

\begin{problem}\label{prob: strong-query}
The hidden set $\hidden$ is a subset of the hypercube $\{0,1\}^d$ of size $n$. The codebreaker is allowed to query $\vecq \in \{0,1,2\}^d$ of their choice in response to which they learn $\min_{\vecx \in \hidden} \disth(\vecx, \vecq)$, i.e., the Hamming distance of $\vecq$ to the point in $\hidden$ nearest to it.
\end{problem}

One motivation for studying such a query model with extra characters is the work of \cite{hu2022recovery}. They study the problem in which a secret binary string in $\{0,1\}^d$ has to be recovered via access to distance oracles for non-decomposable distance metrics such as edit distance, $(p)$-Dynamic Time Warping, and Fr\'echet distances.  They showed that recovery is provably impossible for the exact recovery problem under some distance metrics such as Dynamic Time Warping unless extra characters are allowed in the queries. For recovery using the edit distance oracle, they gave algorithms with better query complexity when queries were allowed to have extra characters. The question of what advantage allowing extra characters gives us in our setting is also of interest.

While our main goal is to design query-efficient algorithms for Problems 1-3, we are also interested in the \textit{adaptivity} of algorithms. The adaptivity of an algorithm is the number of rounds over which the algorithm makes queries and is a measure of its parallelizability. Formally, we have the following definition:

\begin{definition}[$r$-Round Adaptive Algorithm] An algorithm is said to be $r$-round adaptive if it makes queries in $r$ batches $B_1$, \ldots, $B_r$ such that for each $j \in [r]\setminus\{1\}$, the queries made by the algorithm in the $j$-th batch $B_j$ are a function only of the responses to queries in the batches $B_1, \ldots ,B_{j-1}$.
\end{definition}

While it would be ideal to design algorithms that are both query efficient and use few rounds of adaptivity, we show that such algorithms do not exist for  \Cref{prob: hamming-vanilla}. Thus, our primary goal is to understand the tradeoff between the query complexity and adaptivity of algorithms for these problems.

\subsection{Motivation}
One strong motivation for studying these problems is that they seem to be fundamental algorithmic questions and have not been studied previously to the best of our knowledge. Furthermore, several similar questions have interested the community (see \Cref{sec: related}). We now list some areas where the proposed variants of the mastermind problem have applications or are relevant.

\paragraph{Data Reconstruction from a Nearest Neighbor Data Structure.} An obvious application of the generalized mastermind problem is the efficient recovery of data given access to its nearest neighbor data structure. In particular,  a better understanding of these problems may help prevent adversarial attacks from learning private data using a nearest neighbor data structure. 

 \paragraph{Optimization. } 

 \Cref{prob: hamming-vanilla} can be formulated as the algorithmic task of recovering all global minima of the function $f(x) = \disth(x,H)$ whose input is a point $x \in \{0,1\}^d$ and output is the hamming distance of $x$ to the hidden set $H$. Our lower bounds show that recovering all global minima, even in such a structured and simple setting, is hard. 
 
 The key difficulty in designing algorithms for \Cref{prob: hamming-vanilla} stems from a phenomenon akin to ``getting trapped in local minima'' that arises in optimization. In the lower-bound for \Cref{prob: hamming-vanilla} (i.e., \Cref{thm: final-r-round-lower-bound}), the hardness arises from a small set of hidden points $\hidden' \subset H$  preventing the remaining points $\hidden \setminus \hidden'$ from being discovered. More precisely, $H'$ has the property that a uniformly random point from $\{0,1\}^d$ is closer to some point in it than to any point in $\hidden \setminus \hidden'$ with high probability. Therefore, the algorithm has to actively evade points in $\hidden'$ to learn new points. 

     \paragraph{Adverserially Robust Learning. } Deep learning methods, in general, are not robust to adversarial examples \cite{zhang2020adversarial, 9013065, 8294186}. Recent developments such as \cite{lecuyer2019certified, cohen2019certified} have led to methods achieving provable robustness with respect to $\ell_p$-norm perturbations for continuous domains. A standard procedure to adapt these ideas for discrete domains is to pre-process the inputs by first mapping them to $\ell_p$ space before classification. As noted by \cite{hu2022recovery}, efficient non-adaptive algorithms for the exact recovery problems can be used to do exactly this. For instance, a non-adaptive algorithm for the generalized Mastermind problem can be used to losslessly map the input to a $t$-dimensional vector of responses to the $t$ queries made by the algorithm.

\subsection{Our Contributions}
 For \Cref{prob: hamming-vanilla}, we present a two-round adaptive deterministic algorithm which makes $2^{O(\sqrt{d \cdot \log d \cdot \log n})}$ queries to retrieve $\hidden$. On the negative side, we show that any $r$-round adaptive algorithm must make $\exp(\Omega(d^{3^{-(r-1)}}))$ queries even when the input has $O(d^3)$ points. This implies that algorithms with $\poly(d)$ query complexity must use $\Omega(\log \log d)$ rounds of adaptivity. Also, for the special case of non-adaptive algorithms, this corresponds to an $\exp(\Omega(d))$ query lower bound, and for the case of two-round adaptive algorithms, a lower bound of $\exp(\Omega(d^{1/3}))$ queries. For \Cref{prob: unitsphere}, we propose a deterministic algorithm in \Cref{sec: rd-algorithm}. The algorithm makes $O(n^{\floor{d/2}})$ queries over $O(n+d)$ adaptive rounds. Finally, we give a deterministic algorithm for \Cref{prob: strong-query} that makes $O(nd)$ queries in \Cref{sec: strong-query}. This shows that allowing extra characters in the queries of \Cref{prob: hamming-vanilla} makes the problem much easier. A simple information-theoretic lower bound shows that this is essentially optimal; any randomized algorithm that succeeds with constant probability must make $\Omega(\frac{nd}{\log d})$ queries.
 
\subsection{Related Work}\label{sec: related}
Past work has studied several variants of the mastermind problem, sometimes under different names. The work of \cite{fernandez2019query} studied the Mastermind problem in which  the codemaker picks $\vecx \in \{-k, -k + 1, \cdots, k-1, k\}^d$ and the codebreaker tries to infer $\vecx$ by making $\ell_p$ distance queries in which they pick $\vecy \in \{-k, -k + 1, \cdots, k-1, k\}^d$ of their choice and are told $\norm{\vecx - \vecy}_p$. Their techniques give non-adaptive algorithms for separable distance metrics (which are metrics that decompose into coordinate sums) such as $\ell_p$-norms, smooth max, Huber loss, etc. The study of the Mastermind problem with non-separable distance metrics such as edit distance, $p$-Dynamic Time Warping, and Fr\'echet distances was initiated by \cite{hu2022recovery}. Inference problems involving graph metrics have been considered in \cite{jiang2019metric, rodriguez2014strong}.  Permutation-based variants of Mastermind are studied in \cite{afshani2019query,el2020exact}.

The well-known coin-weighing problem is equivalent to the original Mastermind problem with two colors \cite{548511}. In this problem, given $n$ coins, each of whose weight is either $w_0$ or $w_1$, the goal is to determine the weight of each coin in a small number of weighings using a spring balance. Optimal bounds for this problem were achieved by \cite{bshouty2009optimal}. The problem of group testing, introduced by \cite{dorfman1943detection}, is also closely related. In this problem, there is a group of people out of which a disease has infected some. There is a testing procedure that, when queried with a subset of people, can determine either that there is someone in the subset who is infected or that no one is. The goal of the problem is to identify the infected population using the testing procedure as few times as possible. The work of \cite{coja2020optimal} obtains optimal non-adaptive and two-round adaptive group testing algorithms. Several other results for group testing are surveyed in \cite{aldridge2019group}.

\section{Preliminaries}
\subsection{Notation}
\label{sec: notation}
We use $\Ot(.)$ to hide logarithmic factors in the dimension $d$.  The set $\{1, \cdots, l\}$ is denoted by $[l]$. We use $\{0,1, \cdots, k\}^d$ to denote the set of vectors each of whose coordinates is in $[k] \cup \{0\}$. For $\vecx, \vecy \in \{0,1, \cdots, k\}^d$, the Hamming distance between $\vecx$ and $\vecy$, denoted by $\disth(\vecx, \vecy)$,  is the number of coordinates $i \in [d]$ such that $\vecx[i] \neq \vecy[i]$. For a subset $S \subset \hypercube$ and a point $\vecx \in \hypercube$ we use the notation $\disth(\vecx, S)$ to denote $\min_{\vecy \in S} \disth(\vecx, \vecy)$.   The support of $\vecx$, denoted by $\support(\vecx)$,  is the set of non-zero coordinates of $\vecx$ and the Hamming weight of $\vecx$, denoted by $\wt(\vecx)$, is the cardinality of the support. For a set of coordinates $I \subset [d]$, the \emph{restriction} of $\vecx$ to $I$ is the vector in $\{0,1, \cdots, k\}^{|I|}$ obtained by deleting all the coordinates of $\vecx$ except those in $I$. We denote the restriction of $\vecx$ to the set $I$ by $\vecx_I$. For a set $S \subset \R^d$, the convex hull of $S$ will be denoted by $\conv(S)$, and the span of the points in $S$ will be denoted by $\spn{S}$. For a subspace $W$ of $\mathbb{R}^d$, its orthogonal complement is denoted by $W^{\perp}$. We use $\unitsphere$ to denote the unit sphere in $\R^d$.

\subsection{Some Useful Concentration Inequalities}

The binomial distribution with $n$ trials and success probability $p$ is denoted by $\text{Bin}(n,p)$. The following are some standard concentration inequalities for the binomial distribution.

\begin{fact}[Hoeffding's Inequality]\label{clm: bin-concentration} 
Let $Y$ be random variable sampled from $\textup{Bin}(d,1/2)$. For $t > 0$ we have, 
 $$\Pr[\, |Y - d/2| \geq t ] \,\leq\, 2 \cdot \exp(-2t^2/d).$$
\end{fact}

\noindent The following is an \textit{anti-concentration} inequality for binomial random variables.
\begin{fact}[Anti-Concentration Inequality for the Binomial Distribution]\label{clm: bin-anti-concentration} Let $Y$ be random variable sampled from $\textup{Bin}(d,1/2)$. When $0 \leq t \leq d/8$ we have:
 $$\Pr[\, Y \leq \,  d/2 - t ] \,\geq \,\frac{1}{15} \cdot\exp(-16t^2/d).$$
\end{fact}
\noindent The above fact implies the following inequality for the minimum of $m$ independent binomial random variables.
\begin{corollary}\label{clm: min-binomial-concentration}
	Suppose that $Y_1, \ldots, Y_m$ are $m$ independent $\textup{Bin}(d,1/2)$ random variables. Define $Y_{\min}$ to be the minimum of the random variables $Y_i$, i.e., $Y_{\min} = \min\limits_i Y_i$. For $m = 2^{o(d)}$ and $\tfrac{1}{2^m} < \delta < 1$ we have,
    $$ \Pr\left[\,Y_{\min}  \leq \frac{d}{2} - \frac{1}{4} \sqrt{d (\log m - \log(15 \cdot \log (1/\delta))  \,)}\right] \geq 1 - \delta .$$

\end{corollary}

\noindent We now define the hypergeometric distribution and state a corresponding concentration inequality.
\begin{definition}[Hypergeometric Distribution]\label{def: hyper}
Suppose that $d$ balls are drawn randomly without replacement from an urn containing $N$ balls, out of which $K$ are white. The distribution of the random variable $X$, which equals the number of white balls drawn, is denoted by $\hgm(N,K,d)$.
\end{definition}

\begin{fact}[Theorem 4 of \cite{janson2016large}]\label{clm: hypergeometric-concentration}
		Let $Y \sim \hgm(N,K,d)$ and $\mu = \E[Y] = \tfrac{NK}{d}$. For $\eps > 0$ less than a sufficiently small constant, 
		$$\Pr\left[\;\left|Y - \mu \right|  \geq \eps \mu \right] \leq 2\cdot\exp\left(\frac{-\eps ^2\mu }{3}\right).$$
\end{fact}

\section{Results for \texorpdfstring{\Cref{prob: hamming-vanilla}}{Problem 2}}\label{sec: hamming-cube-vanilla}
In this section, we present results for \Cref{prob: hamming-vanilla}. Recall that, in this setting the hidden set $\hidden$ consists of $n$ points from the $d$-dimensional hypercube $\{0,1\}^d$. The query model allows the codebreaker to query $\vecq \in \{0,1\}^d$ to learn the point in $\hidden$ with the smallest Hamming distance to $\vecq$. We use $\query{\vecq}$ to denote this point.

In the following sections, we give upper and lower bounds on the query complexity of \Cref{prob: hamming-vanilla}.
\subsection{A Simple 2-Round Adaptive Algorithm}\label{sec: hamming-cube-algorithm}
We analyze the natural algorithm which makes random queries in the first round and then uses the information revealed to adaptively make queries in the second round.

\paragraph{Overview of \Cref{alg: hamming-nearest-point}. }  In the first round, the algorithm queries  $t = 2^{\Tilde{O}(\sqrt{d \log n})}$ points $\vecy_1,  \ldots, \vecy_t$ uniformly at random from $\hypercube$. In the second round, the algorithm queries all points within Hamming distance $r$ of points discovered in the first round.

\begin{algorithm}[tb]
   \caption{Two-Round Adaptive Algorithm}
   \label{alg: hamming-nearest-point}

\begin{algorithmic}[1]
    \STATE For $t := 2^{O\left( \sqrt{d \cdot \log d \cdot \log (n)} \right)}$, query $t$ points $\vecy_1, \ldots, \vecy_t$ which are sampled independently and uniformly at random from $\{0,1\}^d$. Let $\vecz_1, \ldots, \vecz_t$ be the responses to the queries where $\vecz_i = \query{\vecy_i}$. 
    \STATE For each $\vecz_i$, query all points at Hamming distance at most $r$ from it where $r = O(\sqrt{(d\log n)/ \log d})$.
    \STATE Output the set of all points discovered by the queries made by the algorithm.
\end{algorithmic}
\end{algorithm}

    \begin{theorem}\label{thm: hamming-nearest-point}
   Let $H$ be a hidden set with $n \leq 2^{o(d/\log d)}$ points. ~\Cref{alg: hamming-nearest-point}  makes at most $ 2^{O(\sqrt{d \cdot \log d \cdot \log  n})}$ queries and recovers $\hidden$ with probability at least $2/3$.
    \end{theorem}

    Our main lemma is to show that, for any fixed $\vecx \in H$, the first round recovers a point $\vecz \in H$ such that $\disth(\vecx,\vecz) \leq r = \tilde{O}(\sqrt{d \log n})$ with probability at least $(1 - 1/3n)$. It then follows by a union bound that, with probability at least $2/3$, each point in the hidden set is at a distance at most $r$ from some point recovered in round $1$. Thus, the algorithm recovers all points after the second round of queries with constant probability.

    \begin{lemma}\label{lem: main-lem-hamming}
        Fix a hidden point $\vecx \in H$ and let $\vecz$ be the nearest point to $\vecx$ among the points $\vecz_i$ learned after the first round of queries. With probability at least $(1 - 1/3n)$, we have $\disth(\vecx, \vecz) \leq r$. 
    \end{lemma}
    \begin{proof}
        Let $\vecy = \argmin_{\vecy_i} \disth(\vecx, \vecy_i)$ be the nearest query to $\vecx$ among the queries made in the first round. Let $\vecx' \in H$ be some arbitrary but fixed point hidden point satisfying $\disth(\vecx, \vecx') > r$ (the lemma trivially follows if no such $\vecx'$ exists). We show that $\disth(\vecx, \vecy) < \disth(\vecx', \vecy)$ holds with probability at least $(1 - 1/3n^2)$.  A union bound then shows that with probability at least $(1 - 1/3n)$, we simultaneously have $\disth(\vecx, \vecy) < \disth(\vecx'', \vecy)$ for all $\vecx'' \in H$ satisfying $\disth(\vecx'' , \vecx) > r$. It then follows that $\disth(\vecx, \query{\vecy}) \leq r$ holds with probability $(1 - 1/3n)$ thus proving the lemma.

        We have reduced our task to showing that $\disth(\vecx, \vecy) < \disth(\vecx', \vecy)$ holds with probability at least $(1 - 1/3n^2)$. We now make some assumptions that simplify the analysis. For the rest of the analysis, we assume that $\vecx$ is the origin (i.e., all the $d$ coordinates of $\vecx$ are $0$). This corresponds to ``shifting the origin to $\vecx$'' by XOR-ing all hypercube points with $\vecx$. This operation preserves distances between any two points, so our analysis will not lose any generality.  

        We now have $\disth(\vecx,\vecy) = \wt(\vecy)$ where $\wt(\vecy) = |\supp{\vecy}|$ is the number of coordinates in the support of $\vecy$. It is also easy to see that
$\disth(\vecx', \vecy) = \wt(\vecx') + \wt(\vecy) -2 I$
    where $I:= |\support(\vecx') \cap \support(\vecy)|$. It follows that   $\disth(\vecx, \vecy) < \disth(\vecx', \vecy)$ is equivalent to  
    \begin{align}\label{eq: goal-i}
        \wt(\vecx') > 2I. 
    \end{align}In the rest of the proof, we focus on showing that \Cref{eq: goal-i} holds with probability at least $(1 - 1/3n^2)$. We now define some useful events that will help us prove this upper bound on $I$. Specifically, we let $G_1$ to be the good event that $\wt(\vecy) = d/2 - \Omega(\sqrt{d \log t})$. This is a good event because we can use the fact that $\vecy$ has a small support to argue that the size $I$ of the common support of $\vecx'$ and $\vecy$. Also, we let $G_2$ be the good event that $I$ is not too much more than its expectation $(\wt(y) \wt(\vecx'))/d$. 
    \begin{itemize}
        \item For an integer $w \in [0, d]$, let $E_w$ denote the event that $\wt(\vecy) = w$.
        \item  Let $G_1$ denote the event that $\wt(\vecy) \in [d/4, d/2 - (\sqrt{d \log t})/8]$. 
        \item   Let $G_2$ denote the event that ${I < \mu + \sqrt{3\mu \log(12n^2)}}$ where $\mu = (\wt(\vecx') \wt(\vecy))/d$.
    \end{itemize}

    Using some simple algebra, we can now show the following claim. 
    \begin{claim}\label{clm: query-bound-hamming}
        Suppose that the number of queries $t$ satisfies $t = 2^{\Omega((d \log n)/r)}$. If the good events $G_1,G_2$ both occur, then $\wt(\vecx') > 2I$. 
    \end{claim}
    \begin{proof}
        If the event $G_1 \cap G_2$ occurs then we have,
\begin{align*}
    2I &< 2\frac{\wt(\vecy)\cdot\wt(\vecx')}{d} + 2\sqrt{\frac{3\wt(\vecy)\cdot \wt(\vecx')\cdot\log (12n^2)}{d}}\\
    &<\frac{2\left(\tfrac{d}{2} - \tfrac{1}{8} \sqrt{d \cdot \log t} \right)\cdot\wt(\vecx')}{d} +2 \sqrt{3 \cdot \frac{\tfrac d2 \cdot \wt(\vecx')\cdot\log (12n^2)}{d}}\\
    &= \wt(\vecx') - \frac 14 \cdot \wt(\vecx')\cdot \sqrt{\frac{\log t}{d}} +\sqrt{ 6 \cdot \wt(\vecx')\cdot\log (12n^2)}.
\end{align*}
Therefore, $\wt(x') > 2I$ holds if
\begin{align*}
\frac 14 \cdot \wt(\vecx')\cdot \sqrt{\frac{\log t}{d}} > \sqrt{ 6 \cdot \wt(\vecx')\cdot\log (12n^2)}. \\
\end{align*}
The above inequality on rearranging is equivalent to $t > \exp(96\,d \,  \log(12 n^2)/ \wt(x'))$. Since $\wt(x') > r$, we conclude that $t > \exp(96 \, d \, \log (12 n^2)/r)$ is a sufficient condition for $\wt(\vecx') > 2I$. 
 \Qed{$\,$\Cref{clm: query-bound-hamming}}
 
    \end{proof}

    Note that our choice of parameters $t, r$ in \Cref{alg: hamming-nearest-point} satisfy the first condition of the above claim. Therefore, our task reduces to proving that $G_1 \cap G_2$ holds with probability $(1 - 1/3n^2)$ to finish the proof of \Cref{lem: main-lem-hamming}. This is what we do next.
    
    First, we show that $G_1$ holds with high probability.  We observe that $\wt(\vecy)$ is the minimum of $t$ i.i.d binomial random variables distributed as $\text{Bin}(d, 1/2)$. The claim below then follows from standard concentration inequalities. 

     \begin{claim}\label{clm: prob-g1}
      $\Pr[G_1] \geq (1 - 1/6n^2)$.
    \end{claim}
    \begin{proof}
          The queries $\vecy_1, \ldots, \vecy_t$ are chosen independently uniformly at random, and therefore, the distribution of the random variable $\disth(\vecx, \vecy_i) = \wt(\vecy_i)$ is $ \text{Bin}(d,1/2)$. Since $\vecy$ is the nearest query to $\vecx$ it follows that $\wt(\vecy)$  is the minimum of $t$ independent $\text{Bin}(d, 1/2)$ random variables. If $n = 2^{o(d/ \log d)}$ we have that $t = 2^{O(\sqrt{d \log d \log n})} =2^{o(d)}$. 
  
  Applying \Cref{clm: min-binomial-concentration}, we conclude that 
        $\wt(\vecy) < d/2 - \frac{1}{4}\sqrt{d\cdot (\log t - \log(15(\log 12n^2)))}$ holds with probability at least $(1 - 1/12n^2)$. For sufficiently large $d$ we have, 
        ${ \log(15(\log 12n^2)) \leq 2 \log d \leq \frac{3}{4} \log t }$. Therefore,  $\wt(y) < d/2 - (\sqrt{d \log t})/8$ holds with probability at least $(1 - 1/12n^2)$.

To show the lower bound on $\wt(y)$ we use \Cref{clm: bin-concentration}. The probability that a binomial random variable $\text{Bin}(d, 1/2)$ is less than $d/4$ is at most $2 \exp(-d)$. By a union bound, the probability that the minimum of $t = 2^{o(d)}$  such binomial random variables are less than $d/4$ is at most $\exp(- \Omega(d)) \leq 1/12n^2$. Therefore, $\wt(\vecy) \geq d/4$ with probability at least $(1 - 1/12n^2)$.

A final union bound shows that both the lower and upper bounds on $\wt(\vecy)$ hold with probability at least $(1 - 1/6n^2)$. 
\Qed{$\,$\Cref{clm: prob-g1}}

    \end{proof}

        Next, we want to show that the upper bound on $I$ guaranteed by event $G_2$ also holds with high probability. The key observation is the following: conditioned on $\wt(\vecy) = w$, the distribution of $\vecy$ is uniform over points in $\{0,1\}^d$ with weight $w$. Hence, the distribution of $I$ is hypergeometric with mean $\mu = (w \cdot \wt(\vecx'))/d$. Using tail bounds for the hypergeometric distribution, we then show that $I = \mu + O(\sqrt{\mu \log n})$ holds with high probability.  
    \begin{claim}\label{clm: support-intersection-bound}\label{clm: prob-g2} For any integer $w \in [d/4, d/2]$ we have $\Pr[G_2 \, | \, E_w] \geq (1 - 1/6n^2)$. 
    \end{claim}
    \begin{proof}
        Conditioning on $\wt(\vecy) = w$ (the event $E_w$), by symmetry, the distribution of $\vecy$ is uniform over the points in the hypercube with weight $w$. It follows that conditioned on $E_w$, the random variable $I$ is distributed according to  $\hgm(d, \wt(\vecy), \wt(\vecx'))$ (see \Cref{def: hyper}). The claim then follows by setting $\eps = \sqrt{3 \log(12n^2)/\mu}$ and $\mu =(\wt(\vecy) \cdot \wt(\vecx'))/d$ in \Cref{clm: hypergeometric-concentration}. 
    \Qed{$\,$\Cref{clm: prob-g2}}
    
    \end{proof}
    
     The below claim follows from \Cref{clm: prob-g1}, \Cref{clm: prob-g2}, and the law of total probability. 

     \begin{claim}\label{clm: total-probability}
     $\Pr[G_1 \cap G_2] \geq (1 - 1/3n^2).$
     \end{claim}
     \begin{proof}
         Let $J$ denote the set of integers in the interval $[d/4, d/2 - (\sqrt{d \log t})/8]$. Observe that $G_1 = \bigcup_{w \in J} E_w$. It follows that:
     \begin{align*}
        \Pr[G_1 \cap G_2] &= \sum_{w \in J} \Pr[G_2 \cap E_w] 
        = \sum_{w \in J} \Pr[G_2 \,|\, E_w] \Pr[E_w] \geq (1 - 1/6n^2) \cdot \sum_{w \in J} \Pr[E_w]\\
        &=(1 - 1/6n^2) \Pr[G_1] \geq (1 - 1/6n^2)^2\geq 1 - 1/3n^2. \Qed{\,\Cref{clm: total-probability}}
    \end{align*}
    
     \end{proof}

    This completes the proof of \Cref{lem: main-lem-hamming}
    \end{proof}
        The following lemma bounds the number of queries made by \Cref{alg: hamming-nearest-point}, concluding the proof of \Cref{thm: hamming-nearest-point}.
    \begin{lemma} \Cref{alg: hamming-nearest-point} makes at most $ 2^{O(\sqrt{d \cdot \log d \cdot \log  n})}$ queries.
    \end{lemma}  
    \begin{proof}
    The number of queries in the first round is clearly $2^{O(\sqrt{d \cdot \log d \cdot \log n})}$. In round $2$, the algorithm queries all points at a Hamming distance at most $r = O(\sqrt{(d \log n)/\log d})$ from the points discovered in round $1$. This is at most $t \cdot \sum_{i = 0}^r \binom{d}{i} \leq td^{r+1} = 2^{O(\sqrt{d \cdot \log d \cdot \log  n})}.\qedhere$
    \end{proof}

\subsection{A Lower Bound Against \texorpdfstring{$r$}{r}-Round Adaptive Algorithms}
The main goal of the section is to prove query complexity lower bounds. In particular, we study the trade-off between query complexity and adaptivity.

We say that a randomized algorithm has success probability $p$, if, for each input $H \subset \{0,1\}^d$, it correctly learns all the points in $H$ with probability at least $p$. The main result of this section is a lower bound on the query complexity of $r$-round adaptive randomized algorithms with constant success probability.

\begin{theorem}\label{thm: final-r-round-lower-bound}
 Let $d$ be a sufficiently large integer and $r= O(\log \log d)$ be any positive integer.  Any $r$-round adaptive randomized algorithm for \Cref{prob: hamming-vanilla} with success probability at least $2/3$ must make $\exp({\Omega(d^{3^{-(r-1)}})})$ queries even when the size of the hidden set $H$ is promised to be $O(d^3)$. 
\end{theorem}

Let $\deter(q,r)$ denote the set of all deterministic algorithms that make at most $q$ queries over $r$ adaptive rounds. We shall use the following version of Yao's lemma, which reduces proving randomized query complexity lower bounds to that of designing hard distributions for deterministic algorithms. 

\begin{lemma}[Yao's Principle] \label{lem: yao}Suppose that there exists a distribution $\mathcal{D}$ over instances of \Cref{prob: hamming-vanilla} such that for every deterministic algorithm $\mathcal{A} \in \deter(q,r)$,  we have,  $ \Pr_{H \sim \mathcal{D}} [\mathcal{A} \text{ succeeds on } H] < \delta$; then any $r$-round randomized algorithm which makes at most $q$ queries has success probability at most $\delta$.
\end{lemma}

 This motivates the following definition of input distributions, which are hard for deterministic algorithms.

\begin{definition}[$(d, m, q, r, \delta)$-Hard Distribution] Let $\mathcal{D}$ be a distribution  over subsets of $\{0,1\}^d$ that contain at most $m$ points. Such a distribution $\mathcal{D}$ is called $(d, m, q,r, \delta)$-hard if any algorithm $\mathcal{A} \in \deter(q,r)$ has at most a $\delta$-probability of learning a hidden set drawn from $\mathcal{D}$.
\end{definition}

It follows from Yao's lemma that if a $(d,m,q,r,\delta)$-hard distribution exists, then any $r$-round randomized algorithm with query complexity $q$ has success probability at most $\delta$. Next, we construct a hard distribution for non-adaptive algorithms.
\begin{lemma}[Hard Distribution for $1$-Round Algorithms]\label{lem: base-case}
For some constant $c \in (0,1)$, there exists a distribution $\mathcal{D}_1$ which is $(d, d^2,  2^{cd}, 1, 2^{-\Omega(d)})$-hard.  
\end{lemma}
\begin{proof}
For a point $\vecx \in \{0,1\}^d$ let $N(\vecx,r) \subset \hypercube$ be the set of points at Hamming distance exactly $r$ from $\vecx$.

The hard distribution $\mathcal{D}_1$ generates an input $H$ as follows:
\begin{enumerate}[label = (\roman*)]
    \item Sample a uniformly random point $\vecu$ from $\{0,1\}^d$. 
    \item Add all the points in $N(\vecu,2)$ to  $H$. 
    \item Add a uniformly random subset $S$ of $N(\vecu,1)$ to $H$. 
\end{enumerate}

We now show that for some constant $c \in (0,1)$, any algorithm  in $\deter(2^{cd},1)$  learns $H \sim \mathcal{D}_1$ with probability at most $2^{-\Omega(d)}$. The intuition is that if the number of queries is $2^{cd}$, all the queries of the algorithm will be at a distance greater than $2$ from $u$ with high probability. In this case, the points in $N(u, 2)$ ``block'' the algorithm from learning which subset $S$ of $N(u,1)$ was added to $H$ in step (iii). We formalize this argument below.

Let $\mathcal{A} \in \deter(s,1)$ be a non-adaptive deterministic algorithm that queries the points $\vecy_1, \ldots, \vecy_s$. Let $E$ be the event that for all $i \in [s]$ we have $\disth(\vecy_i, \vecu) \geq 2$. Since $\vecu$ is sampled uniformly at random, for a fixed $\vecy_i$, $\disth(\vecy_i, \vecu) \geq 2$ holds with probability $(1 - (1+d) 2^{-d}) \geq (1 - 2d\cdot 2^{-d})$. A union bound then gives $\Pr[\overline{E}] \geq (1 - 2sd \cdot 2^{-d})$.

Now, observe that if event $E$ occurs, for each query $\vecy_i$ we have $\disth(\vecy_i, N(\vecu,2)) < \disth(\vecy_i, N(\vecu,1))$. Moreover, since all the points in $N(\vecu,2)$ are in $H$, the responses to the queries are entirely determined by the choice of $\vecu$ (and, in particular, are independent of $S$).  It follows that if the event $E$ occurs, then the algorithm has at most a $2^{-d}$ probability of correctly guessing $S$. The success probability of the algorithm, therefore, is at most $\Pr[E] \cdot 2^{-d} + \Pr[\overline{E}] \leq 2^{-d} + 2sd \cdot 2^{-d} = 2^{-\Omega(d)}$ if $s \leq 2^{cd}$ for a sufficiently small constant $c \in (0,1)$. 
\end{proof}

We now show that hard distributions can be constructed recursively: a hard distribution for $(r+1)$ round algorithms can be constructed using a hard distribution for $r$-round algorithms. The new distribution will, however, be over points whose dimension is significantly larger (in fact, it will be approximately the cube of the dimension of the former).
\usetikzlibrary{decorations.pathreplacing}
\usetikzlibrary{arrows.meta}

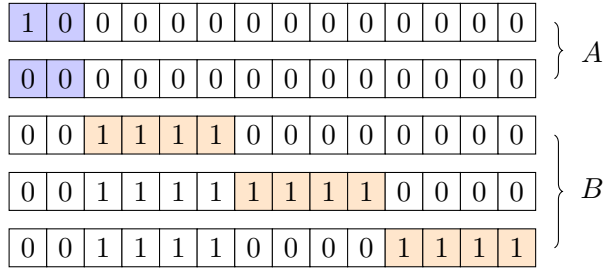
\begin{figure}\label{fig: tikz-lower-bound}

     \centering
  \begin{tikzpicture}
    % Array elements
    \foreach \i/\val/\col in {0/1/1, 1/0/1, 2/0/0, 3/0/0, 4/0/0, 5/0/0, 6/0/0, 7/0/0, 8/0/0, 9/0/0, 10/0/0, 11/0/0, 12/0/0, 13/0/0} {
    \newcommand{\mycolor}{black} % Default color
    \ifcase\col
        % Case 0: Default color
        \renewcommand{\mycolor}{white}
    \or
        % Case 1: Value is 1
        \renewcommand{\mycolor}{blue!20}
    \or
        % Case 2: Value is 2
        \renewcommand{\mycolor}{blue}
    \or
        % Case 3: Value is 3
        \renewcommand{\mycolor}{green}
    \else
        % Default case: Value is not 0, 1, 2, or 3
        \renewcommand{\mycolor}{black}
    \fi
          
      \node[draw,  rectangle, fill = \mycolor, minimum width=0.5cm, minimum height=0.5cm] at (\i*0.5, 0) {\val};
    }

     \foreach \i/\val/\col in {0/0/1, 1/0/1, 2/0/0, 3/0/0, 4/0/0, 5/0/0, 6/0/0, 7/0/0, 8/0/0, 9/0/0, 10/0/0, 11/0/0, 12/0/0, 13/0/0} {
    \newcommand{\mycolor}{black} % Default color
    \ifcase\col
        % Case 0: Default color
        \renewcommand{\mycolor}{white}
    \or
        % Case 1: Value is 1
        \renewcommand{\mycolor}{blue!20}
    \or
        % Case 2: Value is 2
        \renewcommand{\mycolor}{blue}
    \or
        % Case 3: Value is 3
        \renewcommand{\mycolor}{green}
    \else
        % Default case: Value is not 0, 1, 2, or 3
        \renewcommand{\mycolor}{black}
    \fi
          
      \node[draw,  rectangle, fill = \mycolor, minimum width=0.5cm, minimum height=0.5cm] at (\i*0.5, -0.75) {\val};
    }

    \foreach \i/\val/\col in {0/0/0, 1/0/0, 2/1/2, 3/1/2, 4/1/2, 5/1/2, 6/0/0, 7/0/0, 8/0/0, 9/0/0, 10/0/0, 11/0/0, 12/0/0, 13/0/0} {
    \newcommand{\mycolor}{black} % Default color
    \ifcase\col
        % Case 0: Default color
        \renewcommand{\mycolor}{white}
    \or
        % Case 1: Value is 1
        \renewcommand{\mycolor}{blue!20}
    \or
        % Case 2: Value is 2
        \renewcommand{\mycolor}{orange!20}
    \or
        % Case 3: Value is 3
        \renewcommand{\mycolor}{green}
    \else
        % Default case: Value is not 0, 1, 2, or 3
        \renewcommand{\mycolor}{black}
    \fi
          
      \node[draw,  rectangle, fill = \mycolor, minimum width=0.5cm, minimum height=0.5cm] at (\i*0.5, -1.5) {\val};
    }

    \foreach \i/\val/\col in {0/0/0, 1/0/0, 2/1/0, 3/1/0, 4/1/0, 5/1/0, 6/1/2, 7/1/2, 8/1/2, 9/1/2, 10/0/0, 11/0/0, 12/0/0, 13/0/0} {
    \newcommand{\mycolor}{black} % Default color
    \ifcase\col
        % Case 0: Default color
        \renewcommand{\mycolor}{white}
    \or
        % Case 1: Value is 1
        \renewcommand{\mycolor}{blue!20}
    \or
        % Case 2: Value is 2
        \renewcommand{\mycolor}{orange!20}
    \or
        % Case 3: Value is 3
        \renewcommand{\mycolor}{green}
    \else
        % Default case: Value is not 0, 1, 2, or 3
        \renewcommand{\mycolor}{black}
    \fi
          
      \node[draw,  rectangle, fill = \mycolor, minimum width=0.5cm, minimum height=0.5cm] at (\i*0.5, -2.25) {\val};
    }

    \foreach \i/\val/\col in {0/0/0, 1/0/0, 2/1/0, 3/1/0, 4/1/0, 5/1/0, 6/0/0, 7/0/0, 8/0/0, 9/0/0, 10/1/2, 11/1/2, 12/1/2, 13/1/2} {
    \newcommand{\mycolor}{black} % Default color
    \ifcase\col
        % Case 0: Default color
        \renewcommand{\mycolor}{white}
    \or
        % Case 1: Value is 1
        \renewcommand{\mycolor}{blue!20}
    \or
        % Case 2: Value is 2
        \renewcommand{\mycolor}{orange!20}
    \or
        % Case 3: Value is 3
        \renewcommand{\mycolor}{green}
    \else
        % Default case: Value is not 0, 1, 2, or 3
        \renewcommand{\mycolor}{black}
    \fi
          
      \node[draw,  rectangle, fill = \mycolor, minimum width=0.5cm, minimum height=0.5cm] at (\i*0.5, -3) {\val};
    }

       \draw[decorate,decoration={brace,amplitude=3pt}] (7,0) -- (7,-0.75);
        \node at (7.5, -0.4) {$A$};

       \draw[decorate,decoration={brace,amplitude=3pt}] (7,-1.5) -- (7,-3);
       \node at (7.5, -2.2) {$B$};

  \end{tikzpicture}
    \caption{The figure illustrates how a hard distribution for $r$ round algorithms can be used to create a hard distribution for $r+1$ round algorithms.}
\end{figure}

\begin{lemma}[Inductive Step]\label{lem: induction-step}
    Suppose that a $(t, m, q,r,\delta_1)$-hard distribution exists; then a $(t', m', q ,r+1, \delta_1 + \delta_2)$-hard distribution also exists for some $t' \leq 2500t^2(t + \log_2 (q/\delta_2))+t$ and $m' \leq m + 25(t+\log_2(q/\delta_2))$.
\end{lemma}
\begin{proof}

    Let $\mathcal{D}_r$ be the $(t, m, q,r,\delta_1)$-hard distribution.  Let $t' = 2500t^2(t + \log_2 (q/\delta_2))+t$. 
    The hard distribution $\mathcal{D}_{r+1}$ generates an input $H_{r+1} \subset \{0,1\}^{t'} $ as follows (see the illustration above):
    \begin{enumerate}[label = (\roman*)]
        \item Sample $H_r \subset \{0,1\}^t$ from the distribution $\mathcal{D}_r$. 
        \item Let $A \subset \{0,1\}^{t'}$ be the set of points obtained by padding $(t'- t)$ zeros to the end of each point in $H_r$.
        \item Let $\ell = 100 t^2$ and  $m_1 = 25(t + \log_2(q/\delta_2))$. Define $B = \{\vecx_1, \ldots, \vecx_{m_1}\}$ where $\vecx_i$ is given by:
			$$\vecx_i[j] =\begin{cases}
			1 & \text{If  } j \in  [t + (i-1)\cdot \ell +1, \,  t + i\cdot \ell]\\
			0 & \text{Otherwise} \\
            \end{cases}$$
        \item For $\vecu$ sampled uniformly at random from $\{0,1\}^{t'}$ let $A_u = \{\vecu \oplus \vecx \, | \, \vecx \in A \}$ and $B_u = \{\vecu \oplus \vecx \, | \, \vecx \in B \}$ where $\oplus$ denotes the bitwise-XOR operation.
        \item Let $H_{r+1} = A_u \cup B_u$.   
    \end{enumerate}

    It is easy to check that the bounds on the size of $H_{r+1}$ and the dimension $t'$ satisfy the lemma statement.

    Next, we prove that $\mathcal{D}_{r+1}$ is indeed a hard distribution for $(r+1)$-round adaptive algorithms. Consider an algorithm in  $\deter(r+1, q)$ and an input $H_{r+1}$ drawn from $\mathcal{D}_{r+1}$. We show that after the first round of queries, with high probability, the algorithm gains no information about the set $H_r$. However, the input $H_r$ is itself drawn from a distribution that is hard for $r$-round algorithms; thus the algorithm has a low probability of learning $A$ in the remaining $r$ rounds.

     Let $C \subset \{0,1\}^{t'}$ denote the set of $2^t$ points whose support is a subset of the first $t$ coordinates. The main claim we show is that if $\vecz$ is a uniformly random point in $\{0,1\}^{t'}$, then its distance to $B$ is less than its distance to any point in $C$. To prove this claim, we will require the following simple fact.

    	\begin{claim}\label{clm : nearer-condition-nice}
     Let $\vecx$ be a point in $B$ and $\vecy$ be a point in $C$. If $\vecz$ is a point in  $\{0,1\}^{t'}$ satisfying $|\support(\vecx) \cap \support(\vecz)| > \ell/2 + t$ then  $\disth(\vecz,\vecx) < \disth(\vecz,\vecy)$.
            \end{claim}
    \begin{proof}
            Note that we have the following formula for the Hamming distance between $\veca, \vecb \in \{0,1\}^{t'}$: $\disth(\veca,\vecb) = \wt(\veca) + \wt(\vecb) - 2|\support(\veca) \cap \support(\vecb)|$.
            Using this formula and the fact that $|\support(\vecx)| = \ell$ and $0 \leq |\support(\vecy)| \leq t$ we obtain, 
            \begin{align*}
            &\disth(\vecz,\vecx) - \disth(\vecz,\vecy) \\
                &= \wt(\vecx) - 2|\support(\vecx) \cap \support(\vecz)|- \wt(\vecy) + 2|\support(\vecy) \cap \support(\vecz)|  \\
                &< \ell - 2\cdot (\ell/2 + t) - 0 + 2t = 0.\qedhere
        \end{align*} 
    \end{proof}    

    We now prove the main claim of the argument.
    \begin{claim}\label{lem: far-from-ball}
        If $\vecz$ is sampled uniformly at random from $\{0,1\}^{t'}$ then $\disth(\vecz, C) > \disth(\vecz, B)$ with probability at least $1 - \delta_2/q$.
    \end{claim}
    \begin{proof}
        Fix a point  $\vecy \in C$. We shall show that for a uniformly random point $\vecz$, we have $\disth(\vecz,\vecy) \leq \disth(\vecz,B)$ with probability at most $ 2^{-t} q^{-1}\delta_2$. The claim then follows by a union bound over the $|C| = 2^t$ points in $C$.

         Let $\vecx$ be a point in $B$. Since $\vecz$ is chosen uniformly at random from $\{0,1\}^{t'}$ and $\wt(\vecx) = \ell$, we have that $|\support(\vecx) \cap \support(\vecz)|$ is a $\text{Bin}(\ell, 1/2)$ random variable. Using  \Cref{clm: bin-anti-concentration}, the probability that $|\support(\vecx) \cap \support(\vecz)| \geq (\ell/2 + 2t)$ is at least $\tfrac{1}{15} \exp(-64t^2/\ell) = \tfrac{1}{15} \exp(-16/25) > 1/30$. It follows from \Cref{clm : nearer-condition-nice} that $\disth(\vecz,\vecx) < \disth(\vecz,\vecy)$ holds with probability at least  $1/30$. Since $B = \{\vecx_1, \ldots, \vecx_{m_1}\}$ consists of points with pairwise disjoint support, the random variables ${|\support(\vecx_i) \cap \support(\vecz)|}$ corresponding to different $i$ are mutually independent. Thus, the probability that $\disth(\vecz,B) > \disth(\vecz,\vecy)$  is at most $ (29/30)^{m_1} <  2^{-t} q^{-1}\delta_2 .\qedhere$
        
    \end{proof}
    We now use the claim above to show that after the first round of queries, with high probability, the algorithm learns no information about the set of points $H_{r}$ sampled in line (i).

    Suppose that the algorithm queries the points $\vecz_1, \ldots, \vecz_s$ in the first round. Let $E$ be the event that $\disth(\vecz_i, C) > \disth(\vecz_i, B_u)$ holds for all $\vecz_i$. Observe that $\disth(\vecz_i, C) > \disth(\vecz_i, B_u)$ holds iff $\disth(\vecz_i \oplus \vecu, C) > \disth(\vecz_i \oplus \vecu, B)$. Since $\vecz \oplus \vecu$ is distributed uniformly over $\{0,1\}^{t'}$, it follows from \Cref{lem: far-from-ball}, a union bound that $\Pr[\overline{E}] \leq \delta_2 s/q \leq \delta_2$.

    Conditioned on event $E$, the responses to the queries in round $1$ are determined by $\vecu$ and specifically are independent of the set $A_u \subset C$ and thus are also independent of the points $H_r$ sampled in line (i). Therefore, if the event $E$ occurs, the conditional distribution of $H_r$ after the first round of queries is stochastically identical to the distribution of $H_r$ prior to the queries.

     Even if the algorithm is informed what $\vecu$ is after the first round of queries, it still has to learn the set of points $H_{r}$ in the remaining $r$ rounds. Since the points in $H_r$ are from a $(t,m,q,r, \delta_1)$-hard distribution, the algorithm learns $H_r$ with probability at most $\delta_1$. 
    Therefore the success probability of the algorithm is bounded by $\Pr[E] \cdot \delta_1 + \Pr[\overline{E}] \leq \delta_1 + \delta_2.\qedhere$
\end{proof}

We now combine \Cref{lem: base-case} and \Cref{lem: induction-step} to complete the inductive proof and show the existence of a certain hard distribution against $r$-round adaptive algorithms. 

\begin{lemma}\label{lem: final-hard-distribution}
For any integer $t \geq c_1$ and an integer $r$ satisfying $1 \leq r \leq  t$  there exists a distribution $\mathcal{D}_r$ which is $((100t)^{3^{r-1}}, t^{3^{r}} , 2^{ct},r, 2/3) $-hard where $c_1$ is a sufficiently large constant and $c$ is a constant in $(0,1)$.
\end{lemma}
\begin{proof}
    We shall prove this by the following induction. For each $i \in [r]$, we show that there is a $(d(i), m(i) , 2^{ct}, i, i/3r)$-hard distribution where \[d(i) \leq 20^{1.5 (3^{i-1}-1)} \cdot t^{3^{i-1}} \text{and } m(i) \leq t^{3^{i}}.\]  The lemma then follows by setting $i = r$, since $d(r) \leq 20^{1.5 (3^{r-1}-1)} \cdot t^{3^{r-1}} \leq (100t)^{3^{r-1}}$.

The bounds on $m(1)$ and $d(1)$ in the base case $i = 1$ follow from \Cref{lem: base-case} which guarantees the existence of $(t, t^2, 2^{ct}, 1, 2^{-\Omega(t)})$-hard distribution. Suppose the claim holds until $i = j$ for some $j \geq 1$. 

We now use \Cref{lem: induction-step} to perform the induction step. Setting $\delta_2 = 1/r \geq 1/ t$ in this lemma, we deduce that there is a $(d(j+1), m', 2^{ct}, j+1, (j+1)/r)$-hard distribution with $d(j+1) \leq 2500 d(j)^2 (d(j) + \log_2(q/\delta_2)) + d(j)$. Noting that $q \leq 2^{t}$ and $1/\delta_2 \leq t \leq 2^t$ we conclude that $\log_2(q/\delta_2) \leq 2t \leq 2 d(j)$. This implies that  \[d(j+1) \leq 2500 d(j)^2 (3 d(j)) + d(j) \leq (20 d(j))^3.\] Using the induction hypothesis for $i = j$, we further conclude that \[d(j+1) \leq (20 \cdot  20^{1.5 (3^{j-1}-1)} \cdot t^{3^{j-1}})^3 \leq 20 ^ {1.5 (3^{j} - 1)} \cdot t^{3^{j}}. \]
Next, we use \Cref{lem: induction-step} to bound $m(j+1)$. We have,
\begin{align*}
m(j+1) \leq m(j) + 25(d(j) + \log_2(q/\delta_2)) 
\leq m(j) + 75 d(j) \leq t^{3^{j}} + 75 (100t)^{3^{j-1}} \leq t^{3^{j+1}}
\end{align*}
where the last inequality holds for $t$ sufficiently large.\qedhere
\end{proof}

We now complete the proof of \Cref{thm: final-r-round-lower-bound}.

\begin{proof}[Proof of \Cref{thm: final-r-round-lower-bound}]
    Pick $d$ sufficiently large and $r = O(\log \log d)$ so that $r \ll d^{3^{-(r-1)}}$.  Setting  $t = (d^{3^{-(r-1)}})/100$ in \Cref{lem: final-hard-distribution} shows the existence of a $(d, d^3 , 2^{{O(d^{3^{-(r-1)}})}},r , 1/3)$-hard distribution. \Cref{thm: final-r-round-lower-bound} then follows by Yao's lemma.\Qed{\Cref{thm: final-r-round-lower-bound}}
    
\end{proof}

\section{A Deterministic Algorithm for \texorpdfstring{\Cref{prob: unitsphere}}{Problem 3}}\label{sec: rd-algorithm}

Let $S^{d-1}$ denote the unit sphere in $\R^d$. In \Cref{prob: unitsphere}, the codebreaker wishes to learn a hidden set $H \subseteq 
S^{d-1}$ consisting of $n$ points $\hidden = \{\vecx_1, \ldots, \vecx_n\}$. The query model allows the codebreaker to query a unit vector $\vecq$ of their choice in response to which they learn a point $\vecx$ in  $\hidden$ which has the least Euclidean distance to $\vecq$ (with ties being broken adversarially). In other words,  $\vecx = \argmin\limits_{\vecx_i \in \hidden} \norm{\vecx_i - \vecq}_2$. Note that since
\[\norm{\vecx-\vecq}_2^2 = \norm{\vecx}_2^2 + \norm{\vecq}_2^2 -2\cdot\langle \vecx,\vecq\rangle = 2  -2\cdot\langle \vecx,\vecq\rangle, \]
the point $\vecx$ also has the maximum inner product with $\vecq$ among points in $\hidden$. We now present a deterministic algorithm that learns the set $\hidden$ using $n^{O(d)}$ queries using $O(n + d)$ rounds of adaptivity.

\noindent\textbf{Overview of \Cref{alg: rd}:} Suppose that the points in $\hidden$ have rank $k$. The algorithm first finds a basis $B$ of $k$ linearly independent points in $\hidden$. This is done by repeatedly querying points in the orthogonal complement of points discovered thus far until no new points are discovered. Henceforth, the algorithm only queries points in $\spn{B}$, effectively reducing the problem to $k$ dimensions. In each iteration, the algorithm constructs a convex hull of the points of $\hidden$ found so far and queries the normal vectors of all the $(k-1)$-dimensional faces of the convex hull thus constructed. If no new points are recovered, the algorithm terminates.

\begin{algorithm}[tb]
   \caption{Deterministic Algorithm for Points on $\unitsphere$.}
   \label{alg: rd}
\begin{algorithmic}[1]
    \STATE Initialize $B := \phi$. \COMMENT{When the loop on line $2$ terminates,  $B$ is a basis of $H$}.
    \REPEAT
    \STATE Find any orthonormal basis $\{\vecv_1, \ldots, \vecv_t\}$ of $\spn{B}^{\perp}$.
    \STATE Query $\vecv_j$ and $-\vecv_j$ for $1 \leq j \leq t$.  Let $C\subseteq \hidden$ be the points learned. 
     \IF{there exists a point $\vecx$ in $C \subseteq H$ such that $\langle \vecx, \vecv_j\rangle \neq 0$ for some $\vecv_j$}
    \STATE $B = B \cup \{\vecx\}$.
    \ENDIF
    \UNTIL{No new point is added to $B$}
       
    \STATE Define $k:= |B|$. \COMMENT{ $k$ is the rank of points in $\hidden$.}

    \STATE Initialize $\hat{\hidden} := B$ 
    
    \COMMENT{When the loop on line $11$ terminates $\hat{\hidden}$ will equal $H$. 
 This loop finds all points in $\hidden - \texttt{}B$.}
    \REPEAT
    \STATE Find all the $(k-1)$ dimensional faces of $\conv(\hat{\hidden})$.
    \STATE Let $D$ be the set of unit vectors in the $\spn{\hat{\hidden}}$ 
     normal to some $(k-1)$-face of the $\chull{\hat{\hidden}}$.
     
     \COMMENT{Note that  each $(k-1)$-face of $\conv(\hat{\hidden})$ has two unit normal vectors.}

     \STATE Let $E \subseteq H$ denote the set of new points learned on querying points in $D$.
     \STATE $\hat{\hidden} = \hat{\hidden} \cup E$.
    \UNTIL{$|E| = 0$}
    
    \STATE Output $\hat{\hidden}$.

\end{algorithmic}
\end{algorithm}

\begin{theorem}\label{thm: unit-vector}
\Cref{alg: rd} recovers the set $\hidden \subset \unitsphere$ containing $n$ points  with rank $k$ while making at most $O(n^{\floor{\frac{k}{2}}+1} + kd )$ queries.
\end{theorem}
\noindent It is easy to see that at the end of the first loop, \Cref{alg: rd}, recovers a basis $B \subseteq \hidden$. We now prove that it recovers all remaining points in the second loop. First, we prove a simple fact.
\begin{fact}\label{fact: intersect}
	Let $T$ be a finite subset of $\unitsphere$; then $\conv(T) \cap \unitsphere = T$.
\end{fact}
\begin{proof}
Let $\vecx$ be a point in $\conv(T)$. We show $\vecx$ has Euclidean norm $1$ only if it is a point in $T$. By definition,  $\vecx = \sum_{i=1}^m \alpha_i u_i$ where $\alpha_i \in [0,1]$ for $1 \leq i \leq m$ and $\sum_{i = 1}^m \alpha_i = 1$. We have by the Cauchy-Schwarz inequality that,
\[
    \norm{\vecx}^2_2 = \sum\limits_{1 \leq i,j\leq m} \alpha_i \alpha_j \langle \vecu_i, \vecu_j \rangle \leq \sum\limits_{1 \leq i,j\leq m} \alpha_i \alpha_j \norm{\vecu_i}_2 \norm{\vecu_j}_2  = \sum\limits_{1 \leq i,j\leq m} \alpha_i \alpha_j =\left(\sum\limits_{i = 1}^m \alpha_i \right)^2 = 1.  
\]
Since the $\vecu_i$ are distinct, the equality above holds only if $\alpha_i \alpha_j = 0$ whenever $i \neq j$. This is only possible if there exists $k$ such that $\alpha_k = 1$ and $\alpha_i = 0$ for $i \neq k$. It therefore follows that if $\vecx \in \conv(T)$ and $\norm{\vecx}_2 = 1$ then $\vecx \in T.\qedhere$
\end{proof}

We now prove the correctness of the algorithm.
\begin{lemma}
	The output $\hat{\hidden}$ of \Cref{alg: rd} is $\hidden$.
\end{lemma}
\begin{proof} Consider the stage of the algorithm when it is executing an iteration of the loop on line $11$. Also, suppose there is an undiscovered point $\vecy \in \hidden \setminus \hat{\hidden}$ at this stage.  By \Cref{fact: intersect},  $\chull{\hat{\hidden}} \cap S^{d-1}= \hat{\hidden}$ and therefore $y$ is not in $\chull{\hat{\hidden}}$. Therefore, there is a $(k-1)$-face of $\chull{B}$ which separates $\vecy$ from all the points in $\hat{\hidden}$. Therefore, the unit vector $\vecm$ normal to this face of $\conv(\hat{\hidden})$  satisfies $\langle \vecm, \vecy \rangle > \langle  \vecm, \vecx\rangle$ for all $\vecx$ in $\hat{\hidden}$. Hence querying the unit vectors corresponding to all $(k-1)$-faces of $\conv(\hat{\hidden})$ (line $14$) recovers a point in $\hidden \setminus \hat{\hidden}$. This proves that the algorithm terminates only after all the points in $H $ are recovered.
\end{proof}
To bound the query complexity of the algorithm, we use the following standard result about the number of facets of the convex hull of $n$ points in $d$ dimensions.
\begin{lemma}[Upper Bound Theorem, (Theorem 5.5.1 in \cite{matousek2013lectures}) ]\label{fact: faces}
	The number of $(d-1)$-faces of the convex hull of $n$ points in $d$ dimensions is $O(n^\floor{\frac{d}{2}})$.
\end{lemma}

\begin{lemma}
	\Cref{alg: rd} makes at most $O(n^{\floor{\frac{k}{2}}+1} + kd )$ queries where $k = \text{rank}(\hidden)$.
\end{lemma}
\begin{proof}
	\Cref{alg: rd} uses $O(kd)$ queries to recover a basis of $\hidden$. Each iteration of the loop beginning in line $11$ queries the two normal vectors to the faces of the convex hull of a set of at most $n$ points in $k$ dimensions. By \Cref{fact: faces}, each iteration makes at most $O(n^\floor{\frac{k}{2}})$ queries. Since there are at most $n$ iterations of this loop, the total number of queries made by the algorithm overall is $O(n^{\floor{\frac{k}{2}}+1} + kd). \qedhere$
\end{proof}

\section{A Near Optimal Algorithm for \texorpdfstring{\Cref{prob: strong-query}}{Problem 1}}\label{sec: strong-query}
In this section, we consider \Cref{prob: strong-query} in which the hidden set $\hidden$ is a subset of $\{0,1\}^d$ of cardinality $n$ and that in each query the codebreaker learns $\min_{\vecx \in \hidden} \disth(\vecx, \vecq)$ for $\vecq \in \{0,1,2\}^d$ of their choice. The codebreaker aims to learn $\hidden$ with as few queries as possible.

For $\vecx \in \{0,1\}^d$ and $I \subset [d]$ let $\vecx_I$ denote the restriction of $\vecx$ to $I$ (see \Cref{sec: notation} for a definition). Observe that the query model defined above is equivalent to the following query model, which is more convenient to work with. In each query, the codebreaker picks $I \subset [d]$ and a vector $\vecq \in \{0,1\}^{|I|}$ and as a response to this query learns $\min_{\vecx \in \hidden} \disth(\vecx_I, \vecq_I)$ which is the minimum distance of $\vecq$ to points in $\hidden$ when restricted to the indices in $I$. We use  $\query{I,\vecq}$ to denote such a query.

\subsection{The Algorithm}
The following section presents a deterministic algorithm for \Cref{prob: strong-query} that finds $\hidden$ using $O(nd)$ queries.  In the following lemma, we show that a point in $\hidden$ with a given prefix can be recovered using $O(d)$ queries. This lemma will be a subroutine in \Cref{alg: strong-query-optimal}.

We shall use the following notation throughout the section: for each  $i \in [d]$ let $C_i = \{1,\ldots, i\}$.

\begin{lemma}\label{lem: coord-descent-strong}Given $\vecx' \in \{0,1\}^i$, with $i$ satisfying $0 \leq i \leq d$,  we can recover a point $\vecx \in \hidden$ with prefix $\vecx'$ (given that such an $\vecx$ exists in $\hidden$) using at most $d$ queries.
\end{lemma}
\begin{proof}
	If $\vecx$ is a point in $\hidden$ with prefix $\vecx'$ then $\query{C_i, \vecx'} = 0$. Let $(\vecx',0)$ denote the point in $\{0,1\}^{i+1}$ obtained by appending a $0$ to $\vecx'$.  The $(i+1)$-th coordinate is revealed by $\query{C_{i+1}, (\vecx',0)}$. If the response to the query is $0$, then $x_{i+1} = 0$; otherwise $x_{i+1} = 1$. Repeating this for all remaining coordinates reveals $\vecx$ using, at most, a total of $d$ queries.
\end{proof}

\noindent\textbf{Overview of \Cref{alg: strong-query-optimal}: }Suppose that we are at a stage when the algorithm has discovered a subset $\hidden^*$ of $\hidden$. It now finds a $p$ such that $p$ is a prefix of some point in $H \setminus H^*$ but is not the prefix of any point in $H^*$.
It finds the prefix $p$ of some point in $\hidden \setminus \hidden^*$ such that $p$ is not the prefix of any point in $\hidden^*$.  It then uses the subroutine of \Cref{lem: coord-descent-strong} to learn a point in $\hidden \setminus H^*$ with this prefix,  thereby learning a previously unknown point in $\hidden$. 

To find such a $p$, the algorithm iterates over all prefixes of points in $\hidden^*$. For each prefix $q$, the algorithm flips the last coordinate of $q$ to obtain $\tilde{q}$. It then checks if $\tilde{q}$ is the prefix of any point in $\hidden^*$, and if this is not true, queries $\tilde{q}$. If the response to the query is $0$, the algorithm has found a valid $p$.

\begin{algorithm}[tb]
   \caption{A Deterministic Algorithm for the Stronger Query Model.}
   \label{alg: strong-query-optimal}
   
    \begin{algorithmic}[1]
    \STATE For $1 \leq i \leq d$, let $C_i$ denote the set of coordinates $\{1, 2, \cdots, i\}$
    
    \STATE Initialize $\hidden^* := \phi$ \COMMENT{$\hidden^*$ is the subset of $\hidden$ learned by the algorithm.}
    
    \STATE Initialize  $N:=\phi$ \COMMENT{$N$ is the ``most recently" learned subset of $\hidden$}
    
    \STATE Use \Cref{lem: coord-descent-strong} to find some point $\vecx_0$  in $\hidden$
    
    \STATE  Add $\vecx_0$ to $N$ and $H^*$.
    
    \WHILE{$N$ is non-empty}
    		\FOR{$\vecx \in \hidden^*$}
    			\FOR{$i$ from $1$ to $d$}
    			 \STATE Let $\negate{\vecx}{i} \in\{0,1\}^i$ be the $i$-th prefix of $\vecx$ but with the $i$-th coordinate negated.
    			 \IF{$\negate{\vecx}{i}$ is not the prefix of any point in $\hidden^*$ \underline{\emph{and}} $\query{C_i, \negate{\vecx}{i}} = 0$}
    			 \STATE Use \Cref{lem: coord-descent-strong} to find a point in $\hidden$ with prefix $\negate{\vecx}{i}$.	 
    			 \ENDIF
                \ENDFOR
            \ENDFOR
    	
    	\STATE Let $N$ be the set of new points learned in this while loop iteration.
    	$\hidden^* \leftarrow \hidden^* \cup N$.
    	
     \ENDWHILE
    \STATE Output $\hidden^*$
    \end{algorithmic}
\end{algorithm}

\begin{theorem}\label{thm: strong-query}
	\Cref{alg: strong-query-optimal} recovers the set $\hidden \subset \{0,1\}^d$ while making at most $O(nd)$ queries.
\end{theorem}
\noindent The following lemma proves the correctness of \Cref{alg: strong-query-optimal}.
\begin{lemma}\label{lem: alg-strong-correctness}
 The set $\hidden^*$ output by \Cref{alg: strong-query-optimal} equals $H$.
\end{lemma}
\begin{proof}
	First, we prove that the algorithm terminates. The while loop ends when no new points are learned in its most recent iteration. Therefore, there are at most $n$ iterations of the while loop, after which the algorithm terminates. 
	Next, we show that the algorithm learns all points in $\hidden$. Suppose that the execution of the algorithm is at the while loop on line $6$ and that there exists $\vecy \in \hidden \setminus \hidden^*$; we show that the algorithm learns a point in $\hidden \setminus \hidden^*$ in this iteration of the while loop. Among points in $\hidden^*$, let $\vecz$ be the point with the largest common prefix with $\vecy$ and the length of this prefix be $l$. Observe that  $\negate{\vecx}{(l+1)}$ is not the prefix of any point in $\hidden^*$ and that $\query{C_{l+1}, \negate{\vecx}{(l+1)}} = 0$ since $\negate{\vecx}{(l+1)}$ is a prefix of $\vecy$. Therefore, we are done if we argue that $z$ is contained in $N$.  This is true since, if $z \in \hidden^* \setminus N$, then a previous iteration of the while loop would have retrieved a point the length of whose common prefix with $\vecy$ would be at least $(l+1)$.	This leads to a contradiction.
\end{proof}
\noindent Next, we bound the query complexity of \Cref{alg: strong-query-optimal}, completing the proof of \Cref{thm: strong-query}.
\begin{lemma}
	\Cref{alg: strong-query-optimal} makes $O(nd)$ queries.
\end{lemma}
\begin{proof}
For each point in $\hidden$, the algorithm makes at most one query for each of its prefixes. Having learned the prefix of some undiscovered point in line $10$, the algorithm uses \Cref{lem: coord-descent-strong} to learn a new point making at most $d$ queries. Therefore, the total number of queries made is $O(nd)$.
\end{proof}
\noindent A simple information-theoretic lower bound shows that \Cref{alg: strong-query-optimal} is almost optimal. 
\begin{theorem}\label{thm: strong-query-lb} Consider a randomized algorithm for \Cref{prob: strong-query} that for any set $H \subset \{0,1\}^d$ containing $n$ points, learns all the points in $H$ with probability at least $2/3$. The worst-case query complexity of such an algorithm is $\Omega\left(\frac{nd}{\log d}\right)$.
\end{theorem}
\begin{proof}
	Let $\mathcal{D}$ denote the uniform distribution over all $n$-subsets of $\{0,1\}^d$. Let $\mathcal{A}$ be a deterministic algorithm that learns an input $\hidden \sim \mathcal{D}$ with success probability at least $2/3$ while making at most $t$ queries. By Yao's principle, it suffices to show that $t = \Omega(\frac{nd}{\log d})$.
	
	  Since $\mathcal{A}$ is deterministic, its output is completely determined by the responses to its queries. For each query, it receives a response of length $O(\log d)$ bits. It, therefore, outputs one of $2^{O(t \log d)}$ subsets of $\{0,1\}^d$ and can be correct on at most $2^{O(t \log d)}$ inputs. Since it succeeds with probability at least $2/3$, we must have $t = \Omega(\frac{1}{\log d} \cdot \log \binom{2^d}{n}) = \Omega(\frac{n d}{\log d}).\qedhere$
\end{proof}

\noindent\textit{Adaptivity.} While \Cref{alg: strong-query-optimal} is almost optimal in query complexity, it requires several rounds of adaptivity. In particular, the algorithm as described in \Cref{alg: strong-query-optimal} can be implemented to use $O(n)$ rounds of adaptivity, i.e., $O(1)$ rounds to discover each new point. Alternatively, one can also get a $d$-round adaptive algorithm with the same query complexity as follows:  for $i = 1, \ldots, d$, in the $i$-th round, the algorithm learns the length $i$ prefix of each point in $H$. Given all length $i$ prefixes, the algorithm can figure out all length $(i+1)$ prefixes using at most $2n$ queries.  Therefore, overall, we can achieve a query complexity of $O(nd)$ using $O(\min(n,d))$ adaptive rounds.

It would be interesting if query-efficient algorithms that use fewer rounds of adaptivity exist. Efficient non-adaptive algorithms for \Cref{prob: strong-query} can be ruled out via an argument almost identical to \Cref{lem: base-case}.
\raggedbottom
\begin{theorem}
Consider a non-adaptive randomized algorithm for \Cref{prob: strong-query} that for any set $H \subset \{0,1\}^d$ containing $n$ points, learns all the points in $H$ with probability at least $2/3$. Such an algorithm must make $\Omega(2^d)$ queries.
\end{theorem}

\section{Conclusions and Open Problems}
This work introduces the natural extension of the Mastermind problem to the setting when there are multiple secrets. In the context of \Cref{prob: hamming-vanilla},  our work focused on understanding the query complexity of two-round adaptive algorithms. We proposed a $2^{\tilde{O}(\sqrt{d \log n})}$ query upper bound and showed that any algorithm with $\poly(d)$ query complexity must use $\Omega(\log \log d)$ rounds of adaptivity. Designing a polynomial query complexity algorithm remains open. For \Cref{prob: unitsphere} we gave a simple deterministic algorithm with $n^{O(d)}$ query complexity. This query complexity arose because convex polytopes with $n$ vertices in $d$ dimensions can have as many as $n^{O(d)}$ faces. For this problem, we do not know of any non-trivial lower bounds, and it would be nice to make progress in this direction. In the case of \Cref{prob: strong-query}, we have an almost optimal algorithm that recovers the hidden set while making at most $O(nd)$ queries in $O(\min(n,d))$ rounds. It would be interesting to understand if efficient algorithms that use fewer adaptive rounds exist.

Another interesting direction to explore is studying the generalized mastermind problem for other settings, especially non-separable distance metrics such as edit distance, Fr\'echet distance, and dynamic time warping.

\section*{Acknowledgements}
We thank Rahul Ilango for the helpful discussions during the initial phases of this work. We would also like to thank Rajiv Gandhi for making the collaboration between the authors possible. 

David P. Woodruff was supported in part by a Simons Investigator Award and NSF Grant No. CCF-2335412.

\bibliographystyle{alpha}
\bibliography{arxiv/general}
\end{document}